\documentclass[12pt,english]{article}

\usepackage[utf8]{inputenc}
\usepackage{graphicx} 
\usepackage{color} 
\usepackage{url}
\usepackage{amsmath}
\usepackage{amsthm}
\usepackage{amssymb}
\usepackage[margin=1in]{geometry}
\usepackage{appendix}
\numberwithin{equation}{section}
\newtheorem{theorem}{Theorem}[section]
\newtheorem*{theorem*}{Theorem}

\newtheorem{lemma}[theorem]{Lemma}
\newtheorem{corollary}[theorem]{Corollary}

\theoremstyle{definition}

\newtheorem{remark}[theorem]{Remark}

\newtheorem{definition}[theorem]{Definition}

\newcommand{\V}{\mathcal{V}}
\newcommand{\RR}{\mathbb{R}}
\newcommand{\ZZ}{\mathbb{Z}}

\newcommand{\E}{\mathbb{E}}

\newcommand{\Aff}{\mathrm{Aff}}
\newcommand{\U}{\mathrm{U}}
\renewcommand{\O}{\mathrm{O}}

\newcommand{\GLV}{\mathrm{GL}(V)}
\newcommand{\AffV}{\mathrm{Aff}(V)}
\newcommand{\OV}{\mathrm{O}(V)}

\newcommand{\B}{\mathcal{B}}
\newcommand{\M}{\mathcal{M}}

\newcommand{\N}{\mathcal{N}}

\renewcommand{\L}{\mathcal{L}}

\newcommand{\R}{\mathbb{R}}
\newcommand{\sort}{\mathbf{sort}}
\newcommand{\rowsort}{\mathbf{rowsort}}

\DeclareMathOperator{\GL}{GL}
\DeclareMathOperator{\Span}{span}

\DeclareMathOperator{\SO}{SO}

\DeclareMathOperator{\real}{Re}

\newcommand{\rev}[1]{{\color{black} #1}}

\begin{document}

\title{A transversality theorem for semi-algebraic sets with application to signal recovery from the second moment and cryo-EM}

\author{Tamir Bendory, Nadav Dym, Dan Edidin, and Arun Suresh}

\date{\today}

\maketitle

\begin{abstract}
Semi-algebraic priors are ubiquitous in signal processing and machine learning. 
Prevalent examples include a) linear models where the signal lies in a low-dimensional subspace; b) sparse models where the signal can be represented by only a few coefficients under a suitable basis; and c) a large family of neural network generative models. 
In this paper, we  
prove a transversality theorem for semi-algebraic sets in orthogonal or unitary representations
of groups: with a suitable dimension bound, a generic translate of any semi-algebraic set is transverse to the orbits of the group action.  
This, in turn,  implies that if a signal lies in a low-dimensional semi-algebraic set, then it can be recovered uniquely from measurements that separate orbits. 

As an application, we consider the implications of the transversality theorem to the problem of recovering signals that are translated by random group actions from their second moment.
As a special case, we discuss cryo-EM. This is a leading technology to constitute the spatial structure of biological molecules, and serves as our prime motivation.
In particular, we derive explicit bounds for recovering a molecular structure from the second moment under a semi-algebraic prior and deduce information-theoretic implications. We also obtain information-theoretic bounds for three additional applications: factoring Gram matrices, multi-reference alignment, and phase retrieval. Finally, we deduce bounds for designing permutation invariant separators in machine learning.

\end{abstract}

\section{Introduction}
\label{sec:introduction}
\paragraph{Motivation: cryo-EM.} Single-particle cryo-electron microscopy (cryo-EM) is an increasingly important technology for reconstructing the 3-D structure of biological molecules, such as proteins, at high resolution. 
In 2023, the number of molecular structures resolved using cryo-EM was ten times more than the number in 2013~\cite{wwpdb2024emdb}.
There are several reasons for the growing dominance of cryo-EM. In contrast to the classical X-ray crystallography technology, a cryo-EM experiment does not require crystallization, so it can capture the molecules in their native states and recover the structure of numerous biological molecules that resist crystallization.  
In addition, it has the potential to elucidate the dynamics of biological molecules~\cite{toader2023methods}. 
Cryo-EM was selected by Nature Methods as the ``Method of the Year 2015'', three of its pioneers were awarded the Nobel Prize in Chemistry 2017, and it was chosen as a ``Method to Watch" in 2022 by Nature Methods~\cite{doerr2022dynamic}.

A typical cryo-EM experiment produces hundreds of thousands of noisy tomographic projections of the sought molecular structure, each taken from an unknown viewing direction; a detailed mathematical model is introduced in Section~\ref{sec:cryoEM}.
\rev{The cryo-EM inverse problem involves estimating the 3-D molecular structure from noisy projections}, while the unknown viewing directions are commonly treated as nuisance variables~\cite{bendory2020single}. 
Importantly, 
all existing cryo-EM reconstruction algorithms build upon prior assumptions on the 3-D structure. For example, there is an important class
of algorithms, which are based on Bayesian techniques~\cite{scheres2012relion,punjani2017cryosparc}.
This paper focuses on the substantial and rich family of semi-algebraic priors. 

\paragraph{Semi-algebraic sets.} A semi-algebraic set $\M \subseteq \RR^\rev{S}$ is a finite union of sets, which are defined by polynomial equality and inequality constraints. \rev{If $\M \subset \R^S, \mathcal{N} \subset \R^T$ are semi-algebraic sets,}
we will say that a map $f \colon \M  \to \mathcal{N}$  is
{\em semi-algebraic} if the graph $\Gamma_f = \{(x,f(x)\} \subset
\M \times \mathcal{N}$ is a semi-algebraic subset of $\R^S \times \R^T$. The image of a semi-algebraic set under a semi-algebraic map is always semi-algebraic. Any semi-algebraic set~$\M$ can be written as a finite union of smooth manifolds, and the dimension of~$\M$ is defined to be the maximal dimension of these manifolds.
The assumption that a signal (e.g., the 3-D molecular structure) lies in a semi-algebraic set is referred to as a \emph{semi-algebraic prior}. This work is motivated by three important special cases of semi-algebraic sets.

\begin{itemize}
    \item \emph{Linear priors:} the assumption that the signal of interest lies in some low-dimensional subspace. Linear priors are ubiquitous in signal processing and machine learning, and they are the \rev{foundation} of popular methods, such as Principal Component Analysis (PCA). Linear models were proven highly effective in various stages of the computational pipeline of cryo-EM data processing~\cite{van1981use,zhao2016fast,weiss2023unsupervised}.  Some (but not all) non-linear manifolds are also semi-algebraic sets; a simple example would be a sphere. 

\item \emph{Sparse priors}: the assumption that many signals can be approximated by only a few coefficients under some basis or frame~\cite{elad2010sparse}.
This assumption was used recently to design a new class of cryo-EM algorithms in case not many samples are available~\cite{bendory2023autocorrelation}.
\item \emph{Deep generative models:} these are based on neural networks of the form
\begin{equation} \label{eq:relu}
x=A_\ell\circ \eta_{\ell-1} \circ A_{\ell-1}\circ \ldots \circ \eta_1 \circ A_1(z),
\end{equation}
where $z$ resides in a low-dimensional (latent) space, the $A_i$'s are affine transformations and the $\eta_i$'s are semi-algebraic activation functions, such as the rectification function (ReLU) that sets the negative entries of an input to zero. 
Generative models have demonstrated promising results in various tasks within the cryo-EM computational pipeline, \rev{particularly} in conformational variability analysis - one of the substantial challenges in the field~\cite{zhong2021cryodrgn,chen2021deep,punjani20233dflex}.

\end{itemize}
\paragraph{A transversality theorem for semi-algebraic sets.}
In cryo-EM, the method of moments seeks to recover an unknown structure from the higher-order moments of the measurements~\cite{bendory2020single,kam1980reconstruction}. For computational and statistical reasons, there is a desire to use moments of minimal order. As proved in \cite{bendory2022sample}, when no priors are imposed, the second moment determines a structure up to the action of an ambiguity group $H$, which is a product of certain orthogonal groups.
In previous works~\cite{bendory2023autocorrelation,bendory2022sample}, the authors study sparsity priors that ensure that a structure can be recovered from its second moment. This means that the prior set $\M$
has the property that for any $x \in \M$ the orbit $Hx$ \rev{(that is, the orbit under the group of the product of orthogonal matrices)} intersects $\M$ only in $x$. In other words, the prior set $\M$ is {\em transverse} to the orbits of the ambiguity group $H$.

In this paper, we derive strong recovery guarantees for cryo-EM structures from the second moment by proving a much more fundamental result.
We prove a transversality theorem for semi-algebraic sets
in orthogonal or unitary representations $V$ of arbitrary compact Lie groups $H$. In particular,
we derive bounds on the
dimension of a generic \rev{translate of a} semi-algebraic set~$\M$, which ensures that it is transverse to the orbits of the group action; i.e., the $H$ orbit of any point $x\in \M$ intersects $\M$ only at $x$.
This implies that if $F$ 
is a measurement function (such as the second moment) that separates $H$-orbits,
then $F$ is one-to-one when restricted to $\M$. This means that the prior knowledge
that $x$ lies in $\M$ ensures that $x$ is determined by the measurement
$F(x)$. 
The main transversality theorems for real orthogonal
representations, Theorem~\ref{thm:affine}
and Theorem~\ref{thm:orthogonal}, are stated in Section~\ref{sec:main_result}, and proved in Section~\ref{sec:proof}.
Since our primary interest is in real representations, we state and prove
the corresponding results for complex unitary representations 
in Appendix~\ref{sec:complex} to avoid cluttering the main body of the text.

\rev{\begin{remark}We note that our use of the term {\em transversality} is
not the usual differential-geometric notion of transversal intersection
of manifolds. Instead, we were inspired by S. Kleiman's celebrated theorem in his paper {\em The transversality of a general translate}~\cite{kleiman1974transversality}. There, he proves  
that if $V, W$ are two subvarieties of a homogeneous space for
some algebraic group $H$, then the general translate of $V$ by an element of $H$ intersects $W$ in the expected dimension. While the setup of our transversality results is quite different from Kleiman's, our main results are in the spirit of Kleiman's classic theorem.
\end{remark}} 

\paragraph{Applications.} 
The abstract setting of the transversality theorem allows us to derive several corollaries 
that span a wide range of mathematical models and scientific applications. 
Our main focus is on recovering signals that are translated by random group actions from their second moment. 
In Section~\ref{sec:signal_recovery}, we introduce the implications of the 
transversality theorem to this problem and discuss two immediate consequences. First, we derive bounds on the problem of recovering a matrix from its Gram matrix, assuming the matrix lies in a low-dimensional semi-algebraic set. 
Second, we derive bounds for signal recovery from Fourier magnitudes, under semi-algebraic priors.  This problem is called the \emph{phase retrieval problem}, and it plays a key role in a variety of applications in optical imaging, signal processing, and X-ray crystallography. 
In Section~\ref{sec:cryoEM} we discuss the implications of the transversality theorems to cryo-EM---our chief motivation. We derive conditions under which a 3-D molecular structure can be recovered from the second moment of the measurements. We show that this implies an improved sample complexity (fewer observations are required for reconstruction) and discuss implications. We also discuss the connection to the multi-reference alignment problem. From a representation theory perspective, cryo-EM and multi-reference alignment can be understood as generalizations of the phase retrieval problem:  
While in the latter each irreducible representation appears with multiplicity one, the cryo-EM and multi-reference alignment models typically involve multiple copies of each irreducible representation. 
In Section~\ref{sec:permutation}
we apply the transversality theorem to a problem of designing 
permutation invariant separators in machine learning. This problem does not involve second moments, but rather other group invariants, namely row-wise sorting. 

\paragraph{Prior work.}
There are many previous works in which transversality results are implicitly
proved for specific group actions and for specific types of semi-algebraic priors, such as linear subspaces. For example, the frame phase retrieval results
of~\cite{balan2006signal} can be rephrased in the language
of this paper as stating that a generic linear subspace of dimension $M$ 
in $\R^N$ with $N > 2M -2$ is transverse to the orbits of
the group $\{\pm 1\}^N$ acting on $\R^N$ by coordinate sign changes. Because
our results apply to any semi-algebraic set, the bounds we obtain
from Theorem~\ref{thm:affine} and Theorem~\ref{thm:orthogonal} are sometimes slightly weaker (differing by a constant) from the transversality
bounds for specific semi-algebraic sets with simple structures, such as linear subspaces. 
The relation of our results with the existing phase retrieval literature is discussed further in Section~\ref{sec:signal_recovery}. 
For the cryo-EM case, discussed in Section~\ref{sec:cryoEM}, our bounds are actually much stronger than those previously obtained in the literature for generic orthonormal bases~\cite{bendory2022sample}. For the problem of permutations, we also obtain much stronger bounds than previous work \cite{balan2006signal,dym2022low},  (though with a different notion of genericity).

\section{A transversality theorem for semi-algebraic sets}\label{sec:main_result}

We consider a representation $V$ of a compact Lie group $H$
and prove results for semi-algebraic subsets of $V$.
Since the representations in the applications we consider are real, we focus on real representations and present the generalizations for complex representations in Appendix~\ref{sec:complex}. 

\paragraph{Semi-algebraic groups.}
A group $\Theta$ which has the structure of a semi-algebraic set in
$\R^N$ is called a {\em semi-algebraic group} if the multiplication
and inverse maps for the group are semi-algebraic maps. A semi-algebraic group acts on a semi-algebraic set $V$ by semi-algebraic
automorphisms if for each $\theta \in \Theta$ the map 
$V \stackrel{\theta \cdot } \to V$, $x \mapsto \theta \cdot x$ is a semi-algebraic automorphism of $V$.

\paragraph{Generic semi-algebraic sets.}
Our notion of a generic semi-algebraic set depends
on a choice of a 
semi-algebraic group 
$\Theta$ of semi-algebraic automorphisms of $V$. 
Hereafter, $\theta \cdot \M$ denotes
the image of $\M$ under the automorphism $\theta$.

\begin{definition}
Given a semi-algebraic group $\Theta$ of automorphisms of
$V$, we say that 
a transversality result holds for \emph{$\Theta$-generic 
semi-algebraic subsets} of a given dimension~$M$ if the following
condition is satisfied:
\begin{quote}
    For any 
semi-algebraic set $\M$ with $\dim \M \leq M$,  the transversality result
holds for $\theta \cdot \M$, where $\theta$ is a generic 
element of the semi-algebraic group $\Theta$. 
\end{quote}
Precisely, we mean that the set of $\theta \in \Theta$ for which the
transversality result {\em does not hold} for $\theta \cdot \M$ has
strictly smaller dimension than $\dim \Theta$. 

\end{definition}
While the machinery we develop works for any semi-algebraic group of automorphisms, our main results, which are proved in Section~\ref{sec:proof}, are presented below for three groups of automorphisms of 
$V$: the group $\GLV$ of invertible linear transformations,  the group $\AffV$ of invertible affine transformations, and the 
group $\O(V)$ of orthogonal transformations of $V$. 
Each choice of group may be more relevant for different applications and priors. 
For example, if we consider a sparsity prior---that is, we assume that
signals have a sparse expansion with respect to a generic orthonormal basis---then the natural notion of generic is $\OV$-generic since a
generic orthonormal basis is obtained by an orthogonal transformation from any fixed coordinate system.
This type of prior was studied for cryo-EM in~\cite{bendory2022sample}. Similarly, in~\cite{edidin2023generic} the prior that the signals
lie in a generic linear subspace was studied for X-ray crystallography. This prior corresponds to
$\GLV$-generic linear subspaces.
On the other hand, if we consider deep generative priors of the form~\eqref{eq:relu}, then a generic affine transformation corresponds to a generic choice of parameters for the last layer of the network, where all previous layers can have fixed parameters.

\paragraph{Statements of the main results.}
Our main results are formulated in terms of two parameters: the dimension of the semi-algebraic set $M$ and the effective dimension of the representation $K$,
defined by the dimension of the representation minus the maximum dimension of the orbits,
\begin{equation} \label{eq:K}
    K=\dim V -k(H),
\end{equation}
where 
$ k(H) = \max_{x \in V} \dim H x$. 
Note that the orbits of $H$ have dimension\footnote{A classical theorem of Chevalley states that any compact Lie group is the set of real points of a complex algebraic group and is, therefore, a real algebraic set. More generally, if $K \subset H$ is a closed (and hence compact) Lie subgroup, then the compact homogeneous space $H/K$ is also a real algebraic set~\cite[Theorem 3]{iwahori1966duality}. It follows that if $V$ is a representation of a compact Lie group $H$, then the orbit $H  x$ of any $x \in V$ is also a real algebraic set, since the orbit is isomorphic to the compact homogeneous space $H/H_x$ where $H_x = \{h \in H| h \cdot x = x\}$.
In particular, when we talk about the {\em dimension} of the orbits,  we are considering them as real algebraic subsets of $V$.}
at most
equal to $\dim V$ so $k(H) \leq \dim V$ and hence $K \geq 0$.
Clearly, the problem gets easier for smaller $M$ (the semi-algebraic set is of low dimension) and larger $K$ (the effective dimension of the representation is larger); we prove that the gap between $K$ and $M$ can be small. 
The theorems are proved in Section~\ref{sec:proof}.

\begin{theorem}[Main theorem, linear and affine transformations]\label{thm:affine}
Let $V$ be an orthogonal representation of a compact Lie group $H$. Let $\M \subseteq V$ be a semi-algebraic set of dimension $M$,
and 
let $K$ be as in~\eqref{eq:K}.
Then, 
\begin{enumerate}
	\item If $K>M$, then for a generic $x$ in a $\GLV$ or $\AffV$-generic semi-algebraic 
 subset $\M$ of dimension $M$,  if $h \cdot x \in \rev{\M}$ for some $h \in H$, then 
 $h \cdot x = \pm x$ (for $\GLV$) or $h \cdot x = x$ (for \rev{$\AffV$}).
  \item If $K>2M$, then for all $x$ in a $\GLV$ or \rev{$\AffV$}-generic semi-algebraic 
 subset $\M$ of dimension $M$, if $h \cdot x \in \M$ for some $h \in H$, then 
 $h \cdot x = \pm x$ (for $\GLV$) or $h \cdot x = x$ (for $\AffV$).
 \end{enumerate}
\end{theorem}
An illustration of this theorem when $V=\RR^2$ is provided in Figure \ref{fig}. Panel (a) depicts an $M=0$ dimensional semialgebraic set $\M$,  which we think of as the generic affine image of an axis-aligned finite grid. We consider the orbit of points in this set (like the point denoted by the green arrow), under the action of $\{-1,1\}^2$ by elementwise multiplication. Since   $2=K>0=2M$, the theorem guarantees that the orbits will \rev{ only have a trivial intersection with the grid (namely, only the point itself is both in its orbit and in the set)}. In panel (b), we consider the same group action, now on a one-dimensional semi-algebraic set, a union of a \rev{ray} and a circle. In this case, $K=2$, which is larger than $M=1$ but not larger than $2M$. Therefore, the theorem guarantees that orbits of generic points (like the point marked by a green arrow) will \rev{not intersect the set non-trivially}, but it is possible that orbits of a measure-zero set of points, like the point denoted by a red arrow, will have a non-trivial intersection with the set. Finally, in panel (c), we consider an $M=1$ dimensional semi-algebraic set, and the orbits of $S^1$, a one-dimensional group. In this case, we have $K=1=M $, so the conditions of the theorem are violated. Indeed, the orbits of all non-zero points have a non-trivial intersection.   

\begin{figure}
\includegraphics[width=\columnwidth]{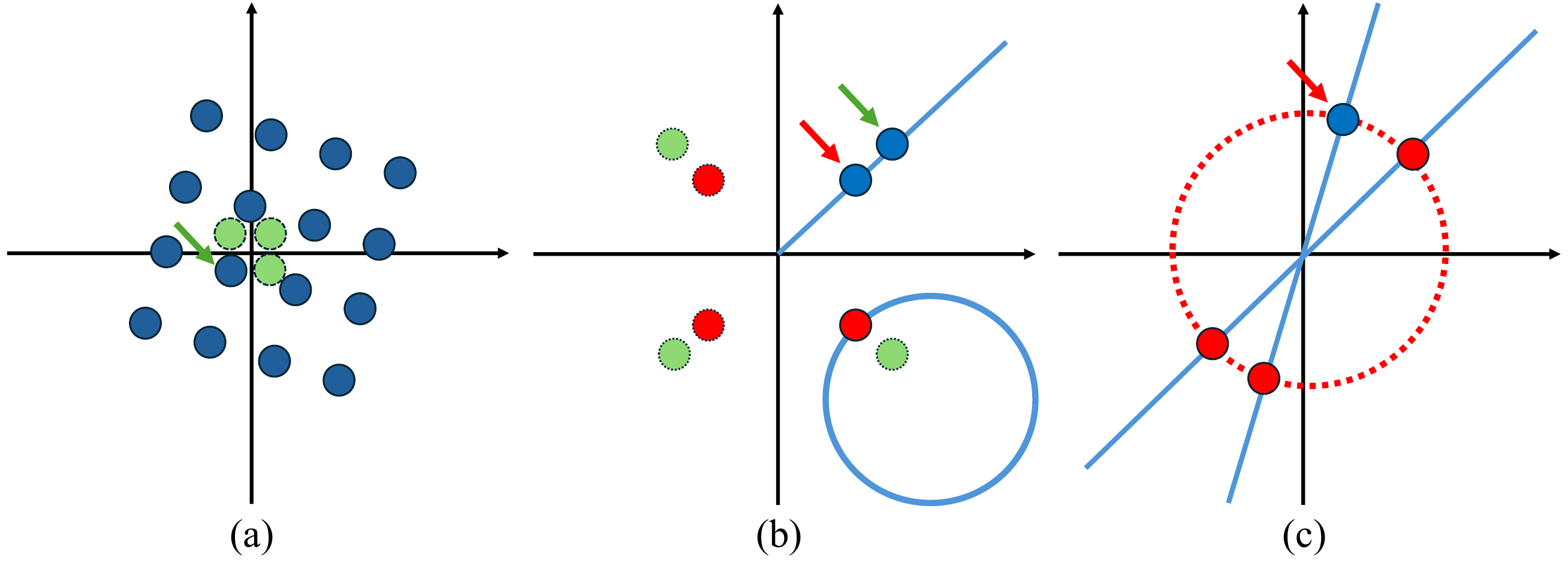}
\label{fig}
\caption{Panels (a)-(c) show three different semi-algebraic sets, colored blue: a zero-dimensional grid in (a), a union of a \rev{ray} and a circle in (b), and a union of two lines in~(c). Each panel selects points in the semi-algebraic set (denoted by red or green arrows) and shows their orbits under the zero-dimensional group $\ZZ_2^2$ (for (a) and (b)) or the one-dimensional group $S^1$ in (c). In (a), the orbits do not have a non-trivial intersection with the semi-algebraic set, in (b) most orbits have only a trivial intersection, but the orbits of a measure zero set of points (like the point denoted by the red arrow) will have a non-trivial intersection. Non-trivial intersections in (c)  are inevitable at all points except for zero. The different behavior in these three cases is determined (generically) by the dimension of the group orbit, the dimension of the semialgebraic set, and the ambient dimension, see Theorem~\ref{thm:affine}. 	
}	
\end{figure}

We now present a similar theorem with respect to generic orthogonal translates of the set, rather than affine or linear translates.

\begin{theorem}[Main theorem, orthogonal transformations]\label{thm:orthogonal}
Let $V$ be a continuous (and hence orthogonal) representation of a compact Lie group $H$. Let $\M \subseteq V$ be a semi-algebraic set of dimension $M$,
and 
let $K$ be as in~\eqref{eq:K}.
Then,
\begin{enumerate}
	\item If  $K>M+2$, then for a generic $x$ in an $\OV$-generic semi algebraic subset $\M$ of dimension $M$, if $h \cdot x \in \M$ for some $h\in H$, then $h \cdot x = \pm x$.
 	\item  If $K>2M+2$, then for all $x$ in an $\OV$-generic 
  semi-algebraic set $\M$ of dimension $M$, if
$h \cdot x \in \M$ for some $h \in H$, then $h \cdot x = \pm  x$.
\end{enumerate}
If, in addition, the orbits of $H$ are connected (for example, if $H$ is connected), then we obtain the improved bounds of 
$K > M+1$ and $K > 2M+1$ respectively. 
\end{theorem}

\paragraph{The dimension of orbits.} 
The results stated above depend on the dimension of the orbit $Hx$.
If $H$ is a finite group, then the dimension of the orbit is always zero. If $H$ is positive-dimensional, then the orbits of $H$ can have dimension strictly smaller than $\dim H$. For example, if
$H = \GL(V)$ acting by linear automorphisms on an $N$-dimensional
(real) vector space $V$, then
$\dim H = N^2$, but non-zero vectors all lie in a single orbit of
dimension equal to $N = \dim V$, and the origin is its own orbit.
If $H = \O(V)$, then $\dim H = \binom{N}{2}$, and the orbits
of non-zero vectors are spheres and thus have dimension $N-1$.

If $H$ is connected, then its orbits are always connected. However, a non-connected group can have all of its orbits connected. 
For example, the orbits of $\O(N)$ acting on $\R^N$ are connected since they are spheres, but $\O(N)$ is not a connected group.

\paragraph{The tightness of the bounds in
Theorem~\ref{thm:affine}.} 
The bound $K>M$ in Theorem~\ref{thm:affine} is sharp, 
since if $H$ is finite, then $K = \dim V$. And if we take
$\M = V$, then clearly the conclusion of the theorem
cannot hold for any translate of $V$ since all translates are equal to
$V$. \rev{Here is an example to show that the bound $K > 2M$ in the second part of the theorem is sharp. Consider the group
$H = \{\pm{1}\}^2$ acting by coordinate-wise sign change. If $\M$ is the parabola 
defined by the equation
$y -x^2=0$, then
a $\GL(V)$-translate of $\M$ is a parabola with equation $ax + by - (c x + dy)^2=0$,  where $ad - bc \neq 0$.
There is a set of $(a,b,c,d)$ of positive measure in $\GLV$ where this parabola contains real points of the form $(x,y), (-x,y)$ with $xy \neq 0$.
For example, taking $(a,b,c,d) = (1,2,2,4)$, the points $(\pm \frac{1}{8}, \frac{1}{16})$
both lie on the parabola. Likewise, a translate of $\M$ by an element of $\AffV$ has equation
$ax + by +A  - (c x + dy+B)^2=0$, where $A, B$ are arbitrary scalars. Once again there is a set of positive measure in $\AffV$ where
the translated parabola has real points of the forms $(x,y), (-x,y)$ with $xy \neq 0$.}

\paragraph{The tightness of the bounds in Theorem~\ref{thm:orthogonal}.}
The following examples show that the bounds
$K> M+2$ (resp. $2M +2$) and $K> M+1$ (resp. $2M+1$) in  Theorem~\ref{thm:orthogonal} are also sharp.
Let $\mathcal{M} = \{(1,0), (-1,0), (0,1), (0,-1)\} \subset\R^2$.
If we take $H = \ZZ_4$, where the generator of $\ZZ_4$ acts by $(x,y) \mapsto (y,-x)$, then $K=2$ and $\dim \mathcal{M} =0$, but any rotation of $\mathcal{M}$ produces four points which lie in a single $H$ orbit.
Similarly, if we take $H = \SO(2)$, then the orbits of $H$ are connected
and one-dimensional so $K=1$, but again any rotation $\mathcal{M}$  consists of points lying in the same $H = \SO(2)$ orbit. 
\rev{Note that for a prior $\M$ of dimension zero transversality holds at a generic point of $\M$ if and only if it holds at all points.}
\rev{\paragraph{The implication of the genericity assumption.} The main results of this paper rely on the assumption that the signal lies in a generic translate of a low-dimensional semi-algebraic set. While this implies that the results hold for almost every such translate, they do not necessarily apply to a fixed given semi-algebraic set. To illustrate this, consider the three primary examples introduced in Section~\ref{sec:introduction}: linear priors, sparsity, and deep generative models.

In the case of deep generative models, the genericity assumption is typically not a major limitation, as it can often be satisfied by appending a generic (e.g., random) layer to a fixed neural network. For the sparsity assumption, if the sparse basis is learned from data---as in dictionary learning~\cite{elad2010sparse}---it is reasonable to expect the model to be generic. However, when sparsity is defined with respect to a fixed, application-driven basis, our results no longer apply. A prominent example is X-ray crystallography, where signals are sparse in the standard basis. Such cases require alternative analytical tools~\cite{elser2018benchmark,bendory2020toward}; see further discussion in Section~\ref{sec:phase_retrieval}.

The same caveat applies to linear priors. If the signal resides in a low-dimensional learned subspace, such as one obtained via PCA, the genericity assumption is likely to hold. In contrast, if the subspace is fixed, as in coherent diffraction imaging~\cite{barnett2022geometry}, our theoretical results do not apply.}

\rev{\paragraph{Context of the transversality theorem.}
As noted above, previous results in phase retrieval, such as frame phase retrieval, can be viewed as transversality theorems for the intersections of linear subspaces with group orbits. To the best of our knowledge, general results like Theorems~\ref{thm:affine}, \ref{thm:orthogonal} have not previously appeared in the math literature. However, our work is inspired by classical moving lemmas, which were proved in intersection theory~\cite{roberts1972moving}. In algebraic geometry,a moving lemma states that if $\alpha$ and $\beta$ are projective cycles\footnote{A cycle of dimension $k$ is a finite formal sum of $k$-dimensional subvarieties.} in the projective space $\mathbb{P}^n$ of dimensions $k$ and
$\ell$, respectively, then the cycle $\alpha$ can be moved in its rational equivalence class to a cycle $\alpha'$ so that $\alpha'$ and $\beta$ intersect in a cycle of the expected dimension $k+\ell -n$ (or
is empty if $k + \ell - n$ is negative).
In the context of group actions,  Kleiman's transversality theorem~\cite{kleiman1974transversality} states that if $V$ and $W$ subvarieties of the orbit $X$ of an algebraic group $H$, then for a general element $\sigma \in H$ the translate $\sigma \cdot V$ intersects $W$ in dimension $\dim V + \dim W - \dim X$, or is empty if this number is negative.

Here, we consider a different problem---that of identifying translates of a semi-algebraic set with a transversality property with respect to the orbits of a compact group $H$. However, the idea of finding a translate of a given set with desired properties is inspired by these older theorems. As is in the case with Kleiman's theorem, we seek to translate by an automorphism of the ambient space, which for us is
the representation $V$ of $H$. The specific results depend on the particular choice of automorphisms, and we have focused on the 
groups $\AffV$, $\GLV$, and $\SO(V)$ because they naturally occur in applications of interest.}

\section{Signal recovery from the second moment} \label{sec:signal_recovery}
Most applications we consider in this paper involve 
signal recovery from the second moment. Thus, we are interested in the transversality of semi-algebraic sets with orbits of
the ambiguity group of the second moment. 
In this section, we first present the model of recovering a signal translated by random group actions from the second moment, and then specialize Theorem~\ref{thm:affine} and Theorem~\ref{thm:orthogonal} to this case. Finally, \rev{we} discuss two important applications: factoring Gram matrices and phase retrieval.

\subsection{Problem formulation}
Consider the problem of recovering a signal $x\in V$ from the second moment of:
\begin{equation} \label{eq:signal_recovery}
    y = g\cdot x, \quad  g\in G,
\end{equation}
where $g$ is a random variable drawn from a uniform (Haar) distribution over a compact group $G$ and $V$ is a finite-dimensional orthogonal representation.
We are interested in determining the $G$-orbit of $x$ from the second moment $\E_{g\sim\text{Unif}(G)} [yy^T]$. 

A general finite-dimensional representation of a compact group
$G$ can be decomposed as 
\begin{equation}
	V = \oplus_{\ell = 1}^L V_\ell^{\oplus R_\ell},
\end{equation}
with the $V_\ell$ are
distinct (non-isomorphic)
irreducible representations of $G$ of dimension $N_\ell$. 
An element of $x \in V$ has a unique $G$-invariant decomposition as a sum
\begin{equation} 
	x = \sum_{\ell = 1}^L \sum_{i= 1}^{R_\ell} x_\ell[i], \label{eq.decomp}
\end{equation} 
where $x_\ell[i]$ 
is in the $i$-th copy of the irreducible representation $V_\ell$.
Conveniently, once a basis for each irreducible representation is fixed, an element of $V$ can be represented by an $L$-tuple
$(X_1, \ldots ,X_L)$, where $X_\ell$ is an 
$N_\ell \times R_\ell$ matrix corresponding to the coefficients of an element in the summand $V_\ell^{\oplus R_\ell}$ according to the given basis.

In~\cite{bendory2022sample}, it was shown that the second moment of~\eqref{eq:signal_recovery} determines the $L$-tuple of $R_\ell \times R_\ell$
symmetric matrices $(X_1^TX_1, \ldots , X_L^T X_L)$.  
	This, in turn, implies that a vector $x$ is determined from the second moment up to the action
	of the product of orthogonal matrices~$H =\prod_{\ell =1}^L \O(N_\ell)$.  
Thus, a prior on $x$ is required to recover a signal uniquely. For example, 
in~\cite{bendory2022sample,ghosh2022sparse,bendory2022sparse,bendory2023autocorrelation,ghosh2023minimax}, it was shown that if the signal is sparse under some basis, then there is a unique sparse signal in the orbit $H$.
Semi-algebraic priors were studied in~\cite{bendory2023phase} for representations composed of irreducible representations with multiplicity one (that is, $R_\ell=1$ for all $\ell$).

 \subsection{Main results for signal recovery from the second moment}
We view $V$ as a representation of $H$, where the $\ell$-th component of $H$, $\O(N_\ell)$, acts diagonally. This means that if we view an element of $V_\ell^{R_\ell}$ as a tuple $v_\ell = (v_\ell[1], v_\ell[2], \ldots , v_\ell[R_\ell])$
 with $v_\ell[r] \in V_\ell$, then $g \cdot v_\ell = (g \cdot v_\ell[1], g \cdot v_\ell[2],
 \ldots g \cdot v_\ell[R_\ell])$. Then,
 the \rev{generic} orbits of $H$ have dimension 
 \begin{equation} \label{eq:k_signal_recovery}
 k = \sum_{\ell =1}^L k_\ell, \quad     k_\ell = \dim \O(N_\ell) - \dim \O(N_\ell - R_\ell),
 \end{equation}
with the understanding that $\dim \O(N_\ell - R_\ell) =0$
if $N_\ell - R_\ell \leq 0$. \footnote{Recall that $\dim \O(N) = N(N-1)/2$.} 
Moreover, the orbits are connected if and only if $\dim V_\ell > 1$ for all $\ell$.
Thus, we derive the following result. 

\begin{corollary}\label{cor:signal_recovery}
Let $V = \oplus_{\ell =1}^L V_\ell^{\oplus R_\ell}$  be a real representation of a compact group $G$. Let $\M \subseteq V$ be a semi-algebraic set of dimension $M$, and  
let $K= \dim V - k,$ where $k$ is given in~\eqref{eq:k_signal_recovery}. 
Then, 
\begin{enumerate}
	\item If $K>M$, then a generic vector $x$ in a $\GLV$ or $\AffV$-generic semi-algebraic set $\M$ of dimension $M$ is determined
 (up to a sign for $\GLV$) by the second moment of~\eqref{eq:signal_recovery}. 
	\item  If $K>2M$, then every $x$ in a $\GLV$ or $\AffV$-generic 
 semi-algebraic set $\M$ of dimension $M$ is determined (up to a sign
 for $\GLV$) by the second moment of~\eqref{eq:signal_recovery}.
	\item If  $K>M+2$, then a generic vector $x$ in an $\OV$-generic semi-algebraic set $\M$ of dimension $M$
 is determined, up to a sign, by the second moment of~\eqref{eq:signal_recovery}.
	\item  If $K>2M+2$, then every vector $x$ in an $\OV$-generic 
 semi-algebraic set $\M$ of dimension $M$ is determined, up to a sign, by the second moment of~\eqref{eq:signal_recovery}.
\end{enumerate}
\end{corollary}

In Section~\ref{sec:cryoEM} we discuss the implications of Corollary~\ref{cor:signal_recovery} to cryo-EM---the prime motivation of this paper. Before we do this, we briefly discuss two additional applications: factoring Gram matrices under semi-algebraic priors and phase retrieval.

\subsection{Applications}

\subsubsection{Factoring Gram matrices under semi-algebraic priors}
As a warm-up application, we 
consider the simple case above where $L =1$; i.e., $V = W^{\oplus R}$ with
$W$ an $N$-dimensional irreducible representation of $H$. 
An element $X \in V$ is simply an $N \times R$
matrix, and the second moment is the Gram matrix
$X^TX$. Generally, the Gram matrix determines $X$ up to multiplication by a matrix in $O(N)$. However, Corollary~\ref{cor:signal_recovery} implies
that an $N \times R$ matrix can be determined from its Gram matrix if it lies in a suitably low-dimensional generic semi-algebraic subset
of $\R^{N \times R}$.

With this notation, we have the following result.

\begin{corollary} \label{cor:gram}
Consider the problem of recovering a matrix $X\in\R^{N\times R}$ from its Gram matrix $X^TX$.
\begin{enumerate} 
\item If $X$ is a generic matrix in a $\GLV$ or $\AffV$-generic semi-algebraic subset 
of dimension $M < NR - \left(\dim\O(N) - \dim \O(N-R)\right)$, then
$X$ can be recovered (up to a sign for $\GLV$) from its Gram matrix.
Likewise, if $2M < NR - \left(\dim \O(N) - \dim \O(N-R)\right)$, then every
$X$ in a $\GLV$ or $\AffV$ can be recovered (up to a sign for $\GLV$) from its Gram matrix.

\item If $X$ is a generic matrix in an $\OV$-generic semi-algebraic subset 
of dimension $M < NR - \left(\dim\O(N) - \dim \O(N-R)\right)-2$, then
$X$ can be recovered (up to a sign for $\GLV$) from its Gram matrix.
Likewise, if $2M < NR - \left(\dim \O(N) - \dim \O(N-R)\right)-2$, then every
$X$ in an $\OV$-generic semi-algebraic subset can be recovered up to a sign from its Gram matrix.
\end{enumerate}
\end{corollary}

\subsubsection{Phase retrieval} \label{sec:phase_retrieval}
The phase retrieval problem entails recovering a signal from its power spectrum: the magnitudes of its Fourier transform. 
Clearly, the power spectrum determines the signal up to a product of $N$ phases (unitary matrices of dimension 1) in Fourier space. 
This problem  
has numerous applications in optical imaging, X-ray crystallography and more, see~\cite{shechtman2015phase,bendory2017fourier,fannjiang2020numerics,grohs2020phase}, and references therein. 
Semi-algebraic priors are essential in phase retrieval. 
For example, in X-ray crystallography, the goal is to recover the sparsely-spread atoms that \rev{form} an image~\cite{elser2018benchmark,bendory2020toward,bendory2022algebraic}. 
In addition, deep generative models have been proven effective in different branches of phase retrieval in recent years~\cite{sinha2017lensless,rivenson2018phase,metzler2018prdeep,hand2018phase,deng2020learning}.

Explicitly, the goal in phase retrieval is to recover a signal $x\in\R^N$ 
from its power spectrum, 
\begin{equation} \label{eq:powerspectrum}
    \left(\hat{x}[0]^*\hat{x}[0], \hat{x}[1]^*\hat{x}[1],\ldots, \hat{x}[N-1]^*\hat{x}[N-1]\right),
\end{equation}
where $\hat{x}[i]$ is the $i$-th entry of the discrete Fourier transform
of the signal $x$. The analysis of \cite[Section 6.2]{edidin2023generic} shows that the power spectrum is equivalent to the second moment of $x$ with respect to the action of the dihedral group $D_{2N}$, which acts on $\R^N$ by \rev{circular} translation and reflection.
As a representation of $D_{2N}$, the vector space $\R^N$ decomposes as
a sum of irreducibles as follows.
\begin{enumerate}
\item If $N$ is even, then
$$\R^N = V_0 \oplus V_1 \oplus \ldots  \oplus V_{N/2-1} \oplus V_{N/2},$$
where  $V_0$ and $V_{N/2}$ are one-dimensional and all other summands are distinct two 
dimensional.
\item If $N$ is odd, then
$$\R^N = V_0 \oplus V_1 \oplus \ldots \oplus V_{(N-1)/2},$$
where  $V_0$ is one-dimensional and all other summands are distinct two 
dimensional.
\end{enumerate}
In particular, the dimension of the orbits of $H= \prod \O(V_i)$ is
$N/2-1$ if $N$ is even and $(N-1)/2$ if $N$ is odd.
Thus, we can invoke Corollary~\ref{cor:signal_recovery} with
$K = N/2 +1$ ($N$ even) and $K = (N+1)/2$ ($N$ odd) 
to obtain the following result. 

\begin{corollary} \label{cor:phase_retrieval} 
Consider the phase retrieval problem of recovering a signal $x\in\R^N$ from its power spectrum~\eqref{eq:powerspectrum}. 
\begin{enumerate}
\item
If $\M$ is a 
$\GLV$ or $\AffV$-generic semi-algebraic set of dimension $M$ with $N \geq 2M$, then the generic vector $x \in \M$ is
determined (up to a sign for $\GLV$) from its power spectrum.
Likewise, if $N\geq 4M$, then every vector in $\GLV$ or $\AffV$-generic semi-algebraic set can be recovered (up to
a sign for $\GLV$) from its power spectrum.

\item If $\M$ is an $\OV$-generic semi-algebraic set of
dimension $M$ with $N \geq  2M+4$, then a generic vector $x \in \M$ is determined, up to a sign, from its power spectrum. Likewise, if $N \geq 4M+4$, then every vector in 
an $\OV$-generic semi-algebraic set is determined, up to a sign, from its power spectrum.
\end{enumerate}
\end{corollary}

\paragraph{Previous work.} 
We note that the bounds on $\dim \M$ obtained using Corollary~\ref{cor:signal_recovery} 
for $\GLV$-generic semi-algebraic subsets were previously obtained using a different argument 
in~\cite{bendory2023phase}. In that paper, the authors also give
the slightly better bounds on $\dim M$ for $\OV$-generic semi-algebraic sets
of $N \geq 2M+2$ (for generic signal recovery) and $N \geq 4M+2$
for recovering all signals, using much more elaborate arguments, which take advantage of the fact that each irreducible representation appears with multiplicity one.

In~\cite{edidin2023generic}, the authors considered phase retrieval for sparse vectors in 
generic bases. There, \rev{by taking advantage of the fact that the set of sparse vectors is the union of linear subspaces and} using algebro-geometric arguments, 
the bounds $N \geq 2M-1$ (generic signal recovery) and $N \geq 4M-3$ (all signal recovery) are obtained.

\section{Cryo-EM} \label{sec:cryoEM}
The computational problem of recovering the 3-D molecular structure from cryo-EM measurements can be formulated as estimating a 3-D function $x$ from $n$ observations of the form 
\begin{equation} \label{eq:cryoEM}
    y_i = T(R_{\omega_i}x) + \varepsilon_i,
\end{equation}
where $T$ is a tomographic projection $Tx(z_1,z_2)=\int_{\R}x(z_1,z_2,z_3)dz_3$, $\varepsilon_i$ is a ``noise" term drawn from an i.i.d.\ normal distribution with zero mean and variance $\sigma^2$, and $R_\omega$ rotates the function $x$, where $\omega$ is drawn from  uniform distribution over $\SO(3).$ Importantly, the goal of the cryo-EM inverse problem is only to estimate the unknown function $x$ from $y_1,\ldots,y_n$, while the unknown 3-D rotations $\omega_1,\ldots,\omega_n$ are treated as nuisance variables.

One of the main challenges in processing cryo-EM data sets and analyzing their performance is the extremely high noise level. This is due to the limited number of electrons the microscope can transmit without damaging the biological sample. From an information-theoretic perspective, it was shown that in the high-noise regime $\sigma\to\infty$, the minimal number of samples required for accurate recovery, regardless of any specific algorithm, needs to scale as $n/\sigma^{2d}\to\infty$, where $d$ is the lowest order moment of the observations that determine the sought structure uniquely~\cite{abbe2018estimation}\footnote{This result is true when the dimension of the signal is finite, as we assume in this paper. In the very high-dimensional regime, the sample complexity is governed by other factors~\cite{romanov2021multi}.}.
In~\cite{bandeira2023estimation}, the authors gave computational evidence that 
the third-order moment determines generic signals 
while the second moment does not, implying that the minimal number of observations $n$ for generic signals scales rapidly with the noise level as $\sigma^6$. 
This motivates identifying classes of signals that can be determined from the second moment, and thus with fewer samples. 
For example, recent works showed that if the structure can be represented with only a few coefficients under some basis~\cite{bendory2022sample}, or as a sparse mixture of Gaussians~\cite{bendory2023autocorrelation}, then the signal can be recovered from the second moment.
In this work, we consider the more general family of signals that lie in a semi-algebraic set (including sparse signals as in~\cite{bendory2022sample}) and 
derive conditions under which a molecular structure can be recovered from the second moment, implying that only $n/\sigma^4 \to \infty$ samples are necessary for accurate estimation in the high noise regime.

\subsection{The second moment of cryo-EM} 
Using spherical coordinates $(r, \theta, \phi)$, 
it is typical to model a molecular structure 
$x\in L^2(\R^3)$
by discretizing $x(r, \theta, \phi)$ 
with $R$ samples $r_1, \ldots , r_{R}$, of the radial coordinates
and bandlimiting the corresponding spherical functions $x(r_i, \theta, \phi)$.
This is a standard assumption in the cryo-EM literature, see for example,~\cite {bandeira2020non}.
Mathematically, this means that we approximate the infinite-dimensional representation $L^2(\R^3)$ with
the finite-dimensional representation
$V = \oplus_{\ell =0}^L V_\ell^{\oplus R}$, where
$L$ is the bandlimit, and $V_{\ell}$ is the $(2\ell +1)$-dimensional
irreducible representation of $\SO(3)$, corresponding to harmonic
polynomials of frequency $\ell$.\footnote{For convenience, we consider a constant number of samples at the radial direction, namely, $R_\ell=R$ for all $\ell$, but omitting this assumption does not alter the analysis.}
In this model, an element $x \in V$ can be identified as an $R$-tuple
$x = (x[1], \ldots , x[R])$, where 
\begin{equation} \label{eq.function}
	x(\theta, \varphi)[r] = \sum_{ \ell=0}^L\sum_{m=-\ell}^\ell X_{\ell}^m[r]
	Y_{\ell}^m(\theta, \varphi),
\end{equation}
and $Y_\ell^m(\theta, \varphi)$ are the spherical harmonics basis functions. 
Therefore, the problem of determining a structure reduces to determining the unknown coefficients $X_\ell^m[r]$ in~\eqref{eq.function}~\cite{bendory2020single,bandeira2020non}. It was shown that the second moment of the observations averaged over $\SO(3)$, is given by the Gram matrices~\cite{kam1980reconstruction,bendory2022sample} 
\begin{equation} \label{eq:Bl}
	B_\ell = X_\ell^* X_\ell, \quad  \ell = 0, \ldots L,
\end{equation}
where $X_{\ell} = \left(X_{\ell}^m[r_i]\right)_{m=-\ell, \ldots , \ell , i = 1, \ldots R}$, contains the spherical harmonic coefficients of $x$\footnote{\rev{We note that when the distribution over $\SO(3)$ is non-uniform, the second moment takes a different form that depends explicitly on the distribution; see~\cite{sharon2020method}. Specifically, in this case, the second moment is given by $\int_\omega \rho(\omega)(T(R_\omega \cdot x)) \, (T(R_\omega \cdot x))^* \, d\omega,$ which defines a rich family of invariant functions of total degree three on $R(G) \times V \times V$, where $R(G)$ denotes the regular representation of $G$. The analysis presented in this paper does not address this more general setting.}}.
 Therefore, the second 
	moment
	determines the coefficient matrices~$X_\ell, \, \ell=0,\ldots,L$ up to the action of the ambiguity group
	$\prod_{\ell =0}^{L} U(2\ell +1)$. 
	If we consider functions that are the Fourier
	transforms of real-valued functions on $\R^3$ (which is the case in cryo-EM), then the ambiguity group is $\prod_{\ell =0}^{L} O(2\ell +1)$.

\subsection{Main result for cryo-EM}  

The dimension of the representation $V$ is 
$N=R \sum_{\ell =0}^L (2 \ell +1) =R(L+1)^2.$ 
If we assume that $R \geq 2L+1$, then the orbits of the
action of $H = \prod_{\ell =0}^L O(2\ell +1)$ have full dimension,
so $k(H) =\sum_{\ell=0}^L\binom{2\ell +1}{2} = 
\frac{L(L+1)(4L+5)}{6}\approx\frac{2L^3}{3}.$ 

\begin{theorem} \label{thm:cryoEM}
Consider the cryo-EM model described above, 
where the signal is taken from the representation 
 $V = \oplus_{\ell =0}^{L} V_\ell^{\oplus R}$ of $\SO(3)$ with $R \geq 2L+1$. Let $\M$ be a semi-algebraic subset of dimension $M$ and 
let $$K= \dim V - \dim \prod_{\ell =0}^L \O(2\ell + 1) = 
(L+1)\left(R (L+1) - {L(4L+5)\over{6}}\right)\approx L^2\left(R+\frac{2L}{3}\right).$$
\begin{enumerate}
	\item If $K>M$, then a generic $x$ in a $\GLV$ or $\AffV$-generic semi-algebraic set $\M$ of dimension 
 $M$ is determined (up to a sign for $\GLV$) by its second moment.
Likewise,  if $K>2M$, then every $x$ in a 
 generic $\AffV$ or $\GLV$-semi-algebraic set $\M$ of dimension 
 $M$ is determined by its second moment (up to a sign for $\GLV$).

	\item If  $K>M+2$, then a generic $x$ in a
 $\OV$-generic semi-algebraic set $\M$ of dimension 
 $M$ is determined by its second moment, up to a sign.
Likewise,  $K>2M+2$ then every $x$ in a 
 generic $\OV$-semi-\rev{algebraic} set $\M$ of dimension 
 $M$ is determined by its second moment, up to a sign.

\end{enumerate}
   \end{theorem}

\begin{corollary}[The sample complexity of cryo-EM under semi-algebraic priors] \label{cor:cryo}
    Consider the cryo-EM model~\eqref{eq:cryoEM} in the high noise regime $\sigma\to\infty$. If the 3-D molecular structure~$x$ lies in a low-dimensional semi-algebraic set that satisfies the conditions of Theorem~\ref{thm:cryoEM}, then it can be determined by only $n/\sigma^4\to\infty$ samples. 
\end{corollary}

\rev{We note that the sample complexity results in the MRA literature are asymptotic in nature, focusing primarily on the rate at which the number of observations must grow (in this case, faster than $\sigma^4$). Nevertheless, there is substantial numerical evidence, dating back more than 25 years, that these asymptotic predictions accurately reflect the model’s behavior in the finite-sample regime; see, for example,~\cite{sigworth1998maximum,bendory2017bispectrum,abbe2018estimation}.}

\paragraph{Dimension counting.}
It is instructive to compare the bounds on the dimension of the semi-algebraic set $M$ and the dimension of the molecular structure (the representation) $N$. Assuming $R\approx 2L$, we have
\begin{equation}
    \frac{K}{N}=\frac{R(L+1)^2 - (L+1){(L+1)L(4L+5)\over{6}}}{R(L+1)^2}\approx\frac{2}{3}.
\end{equation}
This implies that for $\frac{M}{N} \approx \frac{2}{3}$ we can recover
a generic signal, and for $\frac{M}{N} \approx \frac{1}{3}$
our theorem implies that we can recover {\em every} signal
in a generic semi-algebraic set.
Therefore, our results are optimal, possibly up to a constant.

\paragraph{Sparse priors for cryo-EM.}
In \cite{bendory2022sample}, the authors studied the cryo-EM model for signals having a sparse expansion with respect to a generic orthonormal basis. Specifically, the authors proved that if
$\V$ is a generic orthonormal basis, and the expansion
of $x\in\RR^N$ with respect to this basis is $M$-sparse (meaning at most $M$ basis coefficients are non-zero) with $\frac{M}{N}\approx\frac{1}{3}$, then a generic $x$ is determined (up to a sign) from its second moment. 
The set of $M$-sparse vectors with respect to a fixed orthonormal basis is a union of $\binom{N}{M}$ $M$-dimensional linear subspaces, and
the set of $M$-sparse vectors with respect to a generic orthonormal basis $\V$ is an example of an $M$-dimensional $\OV$-generic semi-algebraic set. Notably, the bounds from Theorem~\ref{thm:cryoEM}  substantially beat the bound obtained in~\cite{bendory2022sample}.
Indeed, for $\frac{M}{N} \approx \frac{2}{3}$ we can recover
a generic $M$-sparse signal, and for $\frac{M}{N} \approx \frac{1}{3}$
our theorem implies that we can recover {\em every} $M$-sparse signal
in a generic orthonormal basis.

\paragraph{Potential implications to conformational variability analysis in cryo-EM.}
One of the prime motivations behind the cryo-EM technology is its potential to elucidate multiple conformations (shapes) of biological molecules~\cite{bendory2020single}. 
Different conformations are associated with different functions of the molecule. 
Recovering the different conformations, referred to as the \emph{heterogeneity problem}, is extremely important, giving rise to the development of a wide variety of approaches and techniques. 
The heterogeneity problem is typically classified into two categories: discrete heterogeneity and continuous heterogeneity. In the former, the molecule may appear in several distinct and stable conformations. In the latter and more challenging setup, molecules may present a continuum of conformations, thereby exacerbating the computational challenge~\cite{toader2023methods}.

Our results state that if the 3-D structure lies in a low-dimensional semi-algebraic set, then it can be recovered from fewer observations (Corollary~\ref{cor:cryo}). This might be crucial for the discrete heterogeneity problem, where the total number of measurements is split among the different conformations. \rev{The situation becomes more complicated in the context of continuous heterogeneity}.  This problem is clearly ill-posed without prior on the distribution of conformations, since each tomographic projection (measurement) might be associated with a different conformation.
Our results suggest that semi-algebraic priors on the  \emph{continuous space of conformations} can render the continuous heterogeneity problem well-posed.
Indeed, among the successful techniques proposed in recent years, two outstanding approaches are based on either assuming that the conformations lie on a linear subspace~\cite{anden2018structural,gilles2023bayesian} or can be represented using a deep generative model~\cite{zhong2021cryodrgn,chen2021deep}; both of these are semi-algebraic priors.  
It is important to note, however, that since the prior in the continuous heterogeneity case is not only on the 3-D structure but also on the distribution of conformations, the situation is more subtle than what we consider here.

\subsection{Multi-reference alignment} \label{sec:mra}
The multi-reference alignment model entails estimating a signal $x$ in a representation~$V$ of a compact group~$G$ from $n$ observations of the form
\begin{equation} \label{eq:mra}
    y_i = g_i\cdot x + \varepsilon_i, 
\end{equation}
where $g_1,\ldots,g_n\in G$ are random elements drawn from a uniform distribution over the group~$G$, and $\varepsilon_i\sim\N(0,\sigma^2 I)$ i.i.d. \rev{is the noise.} 
The goal is to estimate $x$ from $y_1,\ldots,y_n$, while the group elements are treated as nuisance variables. 
Evidently, the observations are just noisy realizations of the signal recovery problem~\eqref{eq:signal_recovery}.

The multi-reference alignment model was first suggested as an abstraction of the cryo-EM model, where the tomographic projection is ignored~\cite{bandeira2020non}. Yet, in recent years, this problem has been studied in more generality as a prototype of statistical models with intrinsic algebraic structures, e.g.,~\cite{bendory2017bispectrum,bandeira2023estimation,perry2019sample}.  
As in the cryo-EM case, the sample complexity of the multi-reference alignment model is determined by the lowest order moment that identifies the signal~\cite{abbe2018estimation}, and it was shown that in many cases, the signal is determined only by the third moment~\cite{bandeira2023estimation}. 
Thus, $n/\sigma^6\to\infty$ is a necessary condition for recovery. 
Similar to the cryo-EM problem, this pessimistic result led to a series of papers analyzing when a sparse signal can be recovered from the second moment (implying an improved sample complexity)~\cite{bendory2022sparse,ghosh2022sparse,ghosh2023minimax,bendory2022sample}. This paper generalizes these results as follows. 
Assume we acquire  $n/\sigma^4\to\infty$ samples from~\eqref{eq:mra}. 
In this case,
the empirical second moment of the observations approximates (almost surely) $\E [yy^T]$,  up to a constant term that can be removed assuming the noise level~$\sigma^2$ is known.
Thus, a direct implication of Corollary~\ref{cor:signal_recovery} is the following result on the sample complexity of multi-reference alignment with a generic semi-algebraic prior.

\begin{corollary}[The sample complexity of multi-reference alignment] \label{cor:mra_sample_complexity} 
Consider the multi-reference alignment model~$\eqref{eq:mra}$ and let $\sigma\to\infty$. 
Suppose that $x$ lies in a semi-algebraic set that satisfies the conditions of Corollary~\ref{cor:signal_recovery}. Then, the minimal number of observations required for accurate recovery of $x$, regardless of any specific algorithm, is $n/\sigma^4\to\infty$.
\end{corollary}

We note that the dimension bound on $\M$ for which the improved sample complexity result holds is tighter and more general than the bound given in~\cite{bendory2022sample}.

\section{Permutation invariant separators for machine learning on sets}
\label{sec:permutation}

In this section, we consider the setting where $V=\RR^{d\times n}$ is a representation of the permutation group $S_n$, with respect to the action of applying the same permutation to each of the $d$ rows of a matrix $X\in \RR^{d\times n}$.
This problem is motivated by machine learning applications for objects, like sets \cite{deepsets,qi2017pointnet} and graphs \cite{xu2018how}, which do not come with a natural ordering. 
In recent years, a line of work aims to design permutation invariant ``separators'' for $\RR^{d\times n}$, or a subset $\M \subset \RR^{d\times n}$~\cite{cahill2022groupinvariant,tabaghi2023universal,amir2023neural}.   Namely, a permutation invariant mapping $F:\RR^{d\times n} \to \RR^c$ is separating on $\M$, if for all $X,Y\in \M$ we have that $F(X)=F(Y)$ if and only if $X$ and $Y$ are related by an $S_n$ transformation. For efficiency purposes, it is desirable that the number~$c$, the dimension of \rev{the codomain of $F$}, is as small as possible. Moreover, to be useful for gradient descent-type learning, the function $F$ is typically required to be differentiable, or at least continuous everywhere and piecewise differentiable. 

When $d=1$, one possible solution is choosing $F_1(X)=\sort(X)$, which sorts the elements of the vector $X$ from smallest to largest. This function is a permutation invariant separator and is continuous and piecewise linear.

When $d>1$, there are two common generalizations for the $d=1$ sorting operation. One is lexicographical sorting, where the ordering of $X\in \RR^{d\times n}$ is determined by sorting the first row from small to large, in the case of ties, sorting according to the second row, etc. Lexicographical sorting is a permutation invariant separator, but is not continuous. The second generalization is \rev{sorting each row independently, which we denote by}  $F_d(X)=\rowsort(X) $. This is permutation invariant and continuous, but is not separating. This can be easily seen by noting that 
$F_d(X)=F_d(Y) $ if and only if $Y$ is obtained from $X$ by applying  $d$ different permutations  independently to the $d$ different rows of $X\in \RR^{d\times n}$. Thus, the ``ambiguity group'' of $F_d$ is not $S_n$ but the larger group $(S_n)^d $.
While $F_d$ is not separating on all of $\RR^{d\times n}$, the situation may be different when considering a semialgebraic subset $\M \subseteq \RR^{d\times n}$, since in this case the orbit of a given $X\in \M$ under the $(S_n)^d$ symmetry may not have any intersections with $\M$ other than $X$ itself.  Indeed, we can apply the second part of Theorem \ref{thm:affine}, using the fact that the orbit of the ambiguity group is finite and hence zero-dimensional,  to obtain the following result. 

\begin{corollary}\label{cor:seperators}
Let $\M\subseteq \RR^{d\times n}$ be an $\AffV$-generic semialgebraic set with $\dim(\M)<\frac{1}{2}nd$. Then, the permutation invariant function $F_d(X)=\rowsort(X)$ is 
separating on all of $\M$; i.e., for all $X \in \M$, if
$F(X)=F(Y)$ for some $Y \in \M$ then $X = Y$.
\end{corollary}

This result can be compared with the results of~\cite{dym2022low}, which consider the same group representation with different generic translates. There, the semi-algebraic set $\M$ resides in some space $\RR^{d'\times n}$ (where possibly $d'\neq d $) and the generic translates are multiplications from the left by matrices $B \in \RR^{d\times d'}$. In particular, the authors of~\cite{dym2022low} show that  $\rowsort$ will be invariant and separating on 
$$B \M:=\{BX\,|\, X \in \M\}, $$
if $\dim(\M)<\frac{d}{2} $. 
Thus, the ratio between the dimension $nd$ of the ambient space and the subset $\M$  is $\approx 2n $ in this result, whereas only $\approx 2 $ in Corollary~\ref{cor:seperators}. On the other hand, we note that translates by left matrix multiplication are, in a sense, more natural to this problem as they preserve the permutation structure. In particular, the function $X\mapsto \rowsort(BX) $ is permutation invariant,  while the function $X\mapsto \rowsort(A\cdot X) $ is not. (Here, $A \cdot X $ is any affine transformation $A:\RR^{d\times n} \to \RR^{d\times n} $). 

Of course, it is also possible to apply the first part of Theorem \ref{thm:affine} to obtain a bound of $\dim(\M)<nd $ for a generic $X$. This can be compared with similar bounds obtained in \cite[Proposition 3.7]{balan2022permutation}, for generic translates of vector spaces by matrix multiplication.

\section{Proofs of Theorems~\ref{thm:affine} and \ref{thm:orthogonal}} \label{sec:proof}

\subsection{The fiber lemma}
We prove a general lemma, which is a generalization 
of~\cite[Lemma 2.1]{bendory2023phase}  and also resembles~\cite[Theorem 1.7 and Theorem 3.3]{dym2022low}.
 
\rev{Let $V$ be an orthogonal representation of a compact Lie group $H$} and let
$\Theta$ be a semi-algebraic group of automorphisms of $V$.
For each pair $x,y\in V$, set
$$\Theta(x,y)=\{\theta\in \Theta\,|\, \theta \cdot x \in H (\theta\cdot y) \} .$$ 
In words, $\Theta(x,y)$ is the set of automorphisms $\theta \in V$ with the property that the $\theta$-translate $\theta \cdot x $ lies in the $H$-orbit
of the $\theta$-translate $\theta \cdot y$. Since we assume that $\Theta$ is a group of semi-algebraic automorphisms of $V,$ the set $\Theta(x,y)$ is
a semi-algebraic subset of the group~$\Theta$.

We say that $x \sim y $ if 
$$\Theta(x,y)= \Theta. $$
In this case, if $x,y$ are equivalent in this sense then $\theta \cdot x $ and $\theta \cdot y$ will lie in the same $H$-orbit for any $
\theta \in \Theta$. We note that the equivalence is invariant under the action of $\Theta$ $$x \sim y \text{ if and only if } \theta \cdot x \sim \theta \cdot y,\, \forall \theta \in \Theta, \, x,y \in V.$$

\begin{lemma}[The fiber lemma] \label{lem.AA}
	Let $V$ be a representation of a compact group $H$
 and $\mathcal{M}\subseteq V$ be a semi-algebraic subset of $V$ of dimension $M$. Let $\Theta$ be a semi-algebraic group of automorphisms of $V$.
Assume that 
	$$ \dim(\Theta(x,y))\leq \dim(\Theta)-K, \quad  \forall x,y \in \M, \text{ with }x \not \sim y.$$
	Then, for a generic $\theta \in \Theta$, the following hold:
	\begin{enumerate}
		\item If $K > M$,
		then for a generic vector $z \in \theta \cdot \M$  \rev{if there is a vector $w \in \theta \cdot \M$ such 
        that $z = h \cdot w$}
for some $h \in H$, we must have $z \sim w$. 
		\item If $K > 2M$
  then for any vector $z \in \theta \cdot \M$  \rev{if there is a vector $w \in \theta \cdot \M$ such 
        that $z = h \cdot w$} for some $h \in H$, we must have $z \sim w$. 
  	\end{enumerate}
\end{lemma}
\rev{\begin{remark} We note that Lemma~\ref{lem.AA} holds for any integer
$K$ that satisfies the inequalities 
$$ \dim(\Theta(x,y))\leq \dim(\Theta)-K, \quad  \forall x,y \in \M, \text{ with }x \not \sim y$$ and $K>M$ (resp. $K > 2M$). However, 
when we apply the lemma, we always take $K = \dim V - k(H)$, where $k(H)$ is the maximal dimension of an $H$-orbit, as
defined in~\eqref{eq:K}.
\end{remark}}
\begin{proof}		
	The proof is based on studying the following semi-algebraic incidence set:
	\begin{equation}
		\B=\B(\M,\Theta)=\{(x,y,\theta)\in \M \times \M \times \Theta\,|\,  x\not \sim y \text{ \rev{and} } \theta\cdot x \in H(\theta\cdot y)  \}.
	\end{equation}
	Let $\pi:\M \times \M \times \Theta\to \M \times \M$ be the projection
	$$\pi(x,y,\theta)=(x,y) .$$
	According to ~\cite[Lemma 1.8]{dym2022low}, we can bound the dimension of the incidence $\B$ by  
	\begin{equation}\label{eq:bound}
		\begin{split}
			\dim(\B)&\leq\dim(\pi(\B))+\max_{(x,y)\in \M \times \M, x\not \sim y} \dim \left(\pi^{-1}(x,y) \right)\\ &\leq 2\dim(\M)+ \max_{(x,y)\in \M \times \M, x\not \sim y} \dim \Theta(x,y)\\
			&\leq 2M+ \dim(\Theta)-K.
		\end{split}	
	\end{equation}
		\rev{If $K>2M$, then~\eqref{eq:bound} implies that $\dim(\B)<\dim(\Theta)$. In particular, if
	$\phi$ is the projection $\phi(x,y,\theta)=\theta$, then $\dim(\phi(\B))<\dim(\Theta)$.
	It follows that a generic $\theta\in \Theta$ is not in $\phi(\B)$.  
    Suppose that $z,w\in \theta\cdot \M$ with
    		$z=\theta \cdot x$  and  $w =\theta \cdot y$ for some $x, y \in \M$. If 
$\theta \not \in \phi(B)$, then by definition of $\B$ we know that if $z = h \cdot w$ for some $h \in H$ then $x \sim y$.
Thus, $z = \theta \cdot x$ and $w = \theta \cdot y$ are also equivalent, which completes the proof of the second part of the lemma.}

	For the first part of the lemma, let us assume that $K>M$ so that $\dim(\B)<M+\dim(\Theta)$. We consider two cases: (i) $\dim(\phi(\B))<\dim(\Theta) $ 
 and (ii) $\dim(\phi(\B))=\dim(\Theta)$.
 
 In the first case we have,
	as before, that a generic $\theta \in \Theta$ is not in $\phi(B)$. This implies that for all such $\theta$ if $\theta \cdot x \in H(\theta \cdot y) $ for $x, y \in \M$ then $x \sim y$. Thus, $\theta \cdot x \sim \theta \cdot y$ as well. 
	
	In the second case, we assume that $\dim(\phi(\B))=\dim(\Theta)$
	and we will prove that for a generic $\theta \in \Theta$, the fiber
	$\phi^{-1}(\theta) \cap \B$ has dimension strictly less than $M$. Granted this fact, the proof of the first part proceeds as follows:
	Since $\dim (\phi^{-1}(\theta) \cap \B) < M$, the projection of the fiber onto the $x$ coordinate, which is the set 
	\begin{align}
		\N_\theta&=\{x\in \M\,|\, \exists y \in \M \text{ such that } (x,y,\theta)\in \B\}\\
		&=\{x\in \M\,|\, \exists y \in \M \text{ such that }x \not \sim y \text{ but } (\theta\cdot x) \in H(\theta \cdot y)\}, \nonumber
	\end{align}
	also has a dimension strictly less than $M$.
	Thus, generic $x$ will not be in $\N_\theta$, and for such $x$, if there exists some $y\in \M$ such that $\theta\cdot x \in H(\theta \cdot y) $, then $x \sim y $ and so $\theta \cdot x \sim \theta \cdot y$.
	
	We now prove that the generic fiber of $\phi \colon \B \to \Theta$ has dimension
	less than $M$. 
	By \cite[Proposition 2.9.10]{bochnak1998real}, the semi-algebraic set
	$\Theta$
	is a disjoint union of Nash submanifolds $\Theta_1,\ldots,\Theta_s$. Clearly,
	it suffices to prove that  $\dim (\phi^{-1}(\theta) \cap \B) < M$
	for a generic  $\theta$ in any Nash submanifold $\Theta_\ell$, where $\dim
	\Theta_{\ell} = \dim \Theta$.
	With this observation, we can reduce to the case that $\Theta$ is a
	manifold.
	
	Now we decompose $\B$ as a disjoint union of Nash submanifolds\footnote{\rev{A semi-algebraic subset of $A \subset \R^N$ is called a Nash submanifold of dimension $d$ if,  for every point $x \in A$, there is an open neighborhood $x \in U \subset \R^N$ and a $C^\infty$ differentiable
function $\varphi \colon U \to \R^d$ such that the restriction of $\varphi$ to $U \cap A$ is a diffeomorphism~\cite[Definition 2.9.9]{bochnak1998real}}.}
	$\bigcup_{\ell=1}^L U_\ell$ and 
	let $\L\subseteq \{1,\ldots,L\}$
	denote the subset of indices $\ell$ for which $\phi(U_\ell)$ attains its maximal value $\dim(\Theta)$.
	For every fixed $\ell\in \L$,
	the semi-algebraic version of Sard's
	Theorem~\cite[Theorem 9.6.2.]{bochnak1998real}
	implies that the generic  $\theta \in \phi(U_\ell)$
	is a regular value (i.e., not the image of a critical point) of the restriction of $\phi$ to $U_\ell$.
	By the
	pre-image theorem \cite[Theorem 9.9]{tu2011manifolds},
	every regular value $\theta$ is
	either not in the image of $\phi|_{U_\ell}$, or we have the equality
	\begin{equation}
		\dim(U_\ell)=\dim \phi(U_\ell)+\dim (\phi^{-1}(\theta)\cap U_\ell)=\dim(\Theta) +
		\dim (\phi^{-1}(\theta) \cap U_\ell).    
	\end{equation}
	So,
	\begin{align*}\dim \left( \phi^{-1}(\theta) \cap U_\ell\right)=\dim(U_\ell)-\dim(\Theta)
		\leq \dim(\B)-\dim(\Theta)<M.
	\end{align*}
	The semi-algebraic set, which is the union of
	the images of the critical points of $\{\phi_\ell\}_{\ell \in \L}$ and
	$ \cup_{\ell \not \in \L}\phi(U_\ell)$ has dimension
	strictly less than $\dim \Theta$. If $\theta \in \Theta$ lies in the complement of this set, then we have that the fiber
	\begin{equation}
		\phi^{-1}(\theta)\cap \B=\cup_{\ell\in \L}\left( \phi^{-1}(\theta)\cap U_\ell \right),   
	\end{equation}
	has a dimension strictly smaller than $M$.
\end{proof}

\subsection{Proof of Theorem \ref{thm:affine}}
\paragraph{The case $\Theta = \GLV$.} 
Fix some $x \not \sim y$ in $\M$. By Lemma \ref{lem.AA}, it suffices to show
that the dimension of the  set 
$$\GL(x,y)=\{A \in \GLV\,|\, Ax \in H(Ay) \},$$
is at most $\dim \GLV - K$. If $x$ and $y$ are parallel, then
$Ax$ and $Ay$ are parallel for all $A\in \GLV$. Since $H$ acts
by unitary transformations, $Ax$ is in the $H$-orbit of $Ay$ \rev{only if $|Ax| = |Ay|$ or equivalently only if $|x| = |y|$}. \rev{(Here $|\cdot|$ denotes the usual Euclidean norm.)}  Since
$x$ and $y$ are parallel, this means $x = \pm y$, so $\GL(x,y)$ is empty unless the vectors $x$ and $y$ are $\GLV$-equivalent.

Now, suppose that $x$ and $y$ are linearly independent. Then we can choose an ordered basis $b_1, b_2, \ldots , b_N$ with $b_1 = x$ and
$b_2 = y$. If we express an element $A \in \GLV$ as a matrix with respect to this ordered basis, then $Ax$ is the first \rev{column} of $A$ and
$Ay$ is the second \rev{column} of $A$. The condition that $Ax$ lies in the $H$-orbit of $Ay$ implies that the first \rev{column} of the matrix $A$ lies in
a real algebraic subset of $\R^N$ of dimension at most $k(H)$.
Hence, 
$$\dim\GL(x,y) \leq \dim \GLV - (\dim V - k(H)) = \dim \GLV - K.$$

\paragraph{The case $\Theta = \AffV$.} The argument is
similar to the case $\Theta = \GLV$ but requires more care
because we also consider translations. This time, fix $x \neq y$
in $\M$. It suffices to prove that the dimension of the semi-algebraic  
$$\Aff(x,y)=\{A \in \AffV| \, A(x) \in H(Ay) \},$$
is at most $\dim(\AffV)-K $.  Any affine transformation $A\in \AffV $ can be decomposed as
$$A(x)=L(x+T),$$
where $L:V\to V$ is linear and invertible, and $T\in V$. We  can write $\Aff(x,y)$ as a union of three sets $S_1,S_2,S_3$ defined as
\begin{align*} 
S_1&=\{(L,T)\in \Aff(x,y)| \quad x+T=-(y+T)\}\\ 
S_2&=\{(L,T)\in \Aff(x,y) \quad x+T\neq -(y+T) \text{ but they are linearly dependent}\}\\
S_3&=\{(L,T)\in \Aff(x,y)| \quad x+T \text{ and } y+T \text{ are linearly independent}\}.
\end{align*}
We will prove the bound on the dimension of each of these sets separately.

We start with $S_1$. If  $(L,T)$ is in $S_1$, then there is a unique $T \in V$ such that $x+T = -(y+T)$, namely $T=-\frac{1}{2}(x+y) $. Therefore, the dimension of $S_1$ is at most the dimension of $\GLV$, which gives us a bound on the dimension, 
$$\dim(S_1)=\dim(\GLV)=\dim(\AffV)-\dim(V)\leq \dim(\AffV)-K,$$
since $K=\dim(V)-k(H)\leq \dim(V) $.

Next, we prove that $S_2$ is empty. If $(L,T)$ is in $S_2$, then  $x+T$ and $y+T$ are linearly dependent, but $x+T \neq - ( y+T ) $ by the definition of $S_2$. Also $x+T\neq y+T$ since $x \neq y$. In particular, 
$x+T$ and $-(y+T)$ cannot both be zero. Assume without loss of generality that $x + T \neq 0 $, then there exists some $\lambda\in \RR$ such that $|\lambda|\neq 1$ and $y+T=\lambda(x+T) $.  
 If $L(x+T)=h(L(y+T))$ for some $h \in H$,  then in particular $L(x+T) $ and $L(y+T)=\lambda L(x+T) $ have the same norm because we assume that $H$ acts by orthogonal transformations. But this is impossible since we assume $|\lambda| \neq 1$.
 
 Finally, we consider $S_3$. Fix some  $T$  such that $x+T$ and $y+T$ are linearly independent. We can complete $x+T$ and $y+T$ to an ordered basis $b_1 ,b_2,\ldots, b_N$ of $V$, whose first two elements are $b_1=x+T$ and $b_2=y+T$. An element $L\in \GLV $ can then be parameterized by the values in $V$ it assigns to each of the basis elements. In particular, once a value $L(b_1)$ is specified, we will have that $(L,T)\in \Aff(x,y)$  if and only if $L(b_2) $ is in the $H$-orbit  of $L(b_1)$. Since the dimension of this orbit is $\leq k(H)$, we deduce that 
 $$\dim(S_3)\leq \dim \AffV-\left(\dim(V)-k(H)\right) =\dim(\AffV)-K,$$ 
 which is what we wanted to prove.

\subsection{Proof of Theorem \ref{thm:orthogonal}}

We start by noting that an element $A \in \O(N)$ can be viewed as an $N$-tuple of unit vectors $(w_1,\dots, w_N)$ characterized by the property that $w_k \in \Span(w_1,\dots, w_{k-1})^{\perp}$. Viewed this way, we realize $\O(N)$ as
the last step in a tower of sphere bundles
$$\O(N)=F_N \to F_{N-1} \to \dots \to F_1 = S^{N-1},$$
where $F_k$ consists of all unit vectors $(w_1,\dots, w_k)$ such that $w_k \in \Span(w_1,\dots, w_{k-1})^{\perp}$ and $(w_1,\ldots,w_{k-1})\in F_{k-1}$.

Once again, it is sufficient to bound the dimension of 
$\O(x,y) = \{A \in \OV\,|\, Ax \in H(Ay)\}$. 
If $x$ and $y$ 
are linearly dependent, then $\O(x,y)$ is empty
unless $x = \pm y$ and then $x \sim y$. 
On the other hand, if $x,y$ are linearly independent, then after normalizing, we can assume that $x$ is a unit vector. 
Choose an orthonormal basis $e_1, \ldots , e_N$ for $V$ such that
$e_1 = x$ and $y = ae_1 + be_2$ for some constants $a, b$ with $b \neq 0$  (because $y$ is not parallel to $x$). Now
suppose that $A \in \OV$ is an orthogonal matrix.
The condition that $Ax  = h Ay$ for some $h \in H$ is equivalent to the condition that 
$Ay= aAe_1 + bAe_2$ is in the $H$-orbit of $Ae_1$. 
The vectors $Ae_1$ and $Ae_2$ are just the first two columns of the matrix
$A$. Call these vectors $w_1$ and $w_2$ respectively. 

For a given choice of first column $w_1$, there is an 
$(n-2)$-dimensional set of possible vectors~$w_2$, which can be the second
column of the orthogonal matrix $A$. The condition that
$Ay \in H Ax$
means that we must choose $w_2$ to be in the real 
algebraic subset $-ab^{-1} w_1 +  b^{-1}Hw_1 \subset \R^N$, which has dimension at most $k(H)$. 
Thus, $w_2$ is taken from a 
subset of $S^{n-2}$ of codimension at least~$n-2 - k(H)$, which is equal to $K-2$. On the other hand, once
we have selected the columns $w_1, w_2$, we impose no additional conditions
on the subsequent columns of $A$. Hence, 
$\dim A(x,y) \leq \dim \O(V) -K + 2$.

Now, suppose that the orbits of $H$ are connected. \rev{First,
note that if $\dim Hw_1 = 0$ then $Ay = Ax$ so
$x$ and $y$ are linear dependent and $O(x,y)$ is empty.}
Hence, we may suppose that $Hw_1$ is a connected real algebraic subset
of $\R^N$
of positive dimension. Because orbits are also \rev{smooth} submanifolds,
the connectedness of $Hw_1$ implies that it is also irreducible as an algebraic set. Since $w_2 = -ab^{-1}w_1 + b^{-1}hw_1$ is a unit vector, we impose the additional condition 
that 
$$ 1= \langle w_2, w_2 \rangle = a^2b^{-2} + b^{-2}  -2ab^{-2} \langle w_1, hw_1\rangle.$$
Since $Hw_1$ is positive-dimensional, it contains at least one 
vector $hw_1 \neq \pm w_1$, so the algebraic function
$$Hw_1  \ni hw_1 \mapsto  a^2b^{-2} + b^{-2}  -2ab^{-2} \langle w_1, hw_1\rangle,$$
is not constant on the irreducible algebraic set $Hw_1$ unless $a=0$
and $b= \pm 1$. 
Thus, if $x,y$ are not orthogonal
vectors with the same magnitude, then the condition that $w_2$ is a unit
vector shows that $w_2$ is taken from an algebraic subset of
$S^{n-2}$ of dimension at most $k(H) -1$, so in this 
case $\dim \O(x,y) \leq \dim \O(V) - K+1$.

On the other hand, if $x,y$ are orthogonal vectors of the same magnitude, then after normalizing we must have that 
$w_1 = x$ and $w_2 = \pm y$. The requirement 
that $\langle w_1, w_2 \rangle = 0$ imposes the condition that 
$\langle w_1, h w_1 \rangle=0$. This condition cannot be identically satisfied on
the orbit $Hw_1$,
since it is not satisfied at $w_1 \in Hw_1$. 
Thus, once again we see that $\dim \O(x,y) \leq \dim \O(V) - K +1$.

\section*{Acknowledgments}
T.B. and D.E. were supported by the BSF grant no. 2020159. T.B. is also supported in part by the NSF-BSF grant no. 2019752, and the ISF grant no. 1924/21. D.E. was also supported by NSF-DMS 2205626. N.D. was supported by ISF grant 272/23.

\bibliographystyle{plain}


\appendix 

\section{A transversality theorem for complex representations} \label{sec:complex}
Let $V$ be a complex unitary representation of a compact group $H$.
We now state and prove analogs of Theorem~\ref{thm:affine} and Theorem~\ref{thm:orthogonal} for real semi-algebraic sets, which are generic with respect to the action of the complex algebraic groups $\GLV, \AffV$
and the unitary group $\U(V)$. Note that the term generic
refers to the real Zariski topology. The proofs are similar to the real case, so we only sketch them, indicating the necessary modifications.
Set $$ K= \dim_\R V - k(H)$$
where $k(H) = \max_{x \in V} \dim_\R Hx$.

\begin{theorem}\label{thm:unitary-affine}
Let $V$ be a complex unitary representation of a compact Lie group $H$.
\begin{enumerate}
	\item If $K>M$, then for a generic $x$ in a $\GLV$-generic semi-algebraic subset $\M \subset V$ of dimension $M$ if 
 $h \cdot x  \in \M$ then $h \cdot x = \lambda x$ for some
 $\lambda \in S^1$.
\item If $K > M+1$ then for a generic $x$ in an $\AffV$-generic semi-algebraic subset $\M \subset V$ of dimension $M$ if 
 $h \cdot x  \in \M$ then $h \cdot x = x$.
 \item If $K>2M$, then for all $x$ in a $\GLV$-generic semi-algebraic subset $\M \subset V$ of dimension $M$, 
 if $h \cdot x \in \M$, then $h\cdot x = \lambda x$ for 
 some $\lambda \in S^1$.
\item  If $K>2M+1$, then for all $x$ in an $\AffV$-generic semi-algebraic subset $\M \subset V$ of dimension $M$, 
 if $h \cdot x \in \M$, then $h\cdot x = x$. 
\end{enumerate}
\end{theorem}

\begin{theorem}\label{thm:unitary-unitary}
Let $V$ be a unitary representation of
a compact Lie group $H$.
\begin{enumerate}
	\item If $K>M+3$, then for a generic $x$ in a $\U(V)$-generic
 semi-algebraic subset $\M \subset V$ of dimension $M$,
 if $h \cdot x \in \M$, then $h\cdot x =\lambda x$ for some $\lambda \in S^1$.
	\item  If $K>2M+3$, then for all $x$ in a $\U(V)$-generic semi-algebraic subset $\M \subset V$ of dimension $M$, 
 if $h \cdot x \in \M$, then $h\cdot x =\lambda x$ for some $\lambda \in S^1$.
\end{enumerate}

If in addition the orbits of $H$ are connected (for example if $H$ is connected), then we obtain the improved bounds
$K > M+2$ and $K > 2M+2$ for the generic unitary transformation of~$\mathcal{M}.$ 
\end{theorem}

\begin{proof}[Sketch of the proof of Theorem~\ref{thm:unitary-affine} and
Theorem~\ref{thm:unitary-unitary}]
We note that Lemma~\ref{lem.AA} applies in the complex case if we identify the complex vector space $V$ with $\R^{2 \dim V}$. So we are again reduced to bounding 
$\dim \Theta(x,y) = \{ \theta \in \Theta\,|\, \theta\cdot x \in H(\theta \cdot y)\}$ over all non-equivalent pairs $x,y$,  where $\Theta =\Aff(V)$ or $\Theta = U(V)$.

When $\Theta = \GLV$, then two vectors $x, y$ are equivalent if and only if
$y = e^{\iota \theta}x$ 
and as in the
proof of Theorem~\ref{thm:affine}, $\Theta(x,y)$ is empty unless they have the same magnitude, in which case they are equivalent. Likewise,
if $x,y$ are linearly independent, then we can take them to be the first
two vectors in an ordered basis for $V$. If we represent
$A \in \GLV$ as a matrix with respect to this ordered basis, the condition that $Ax = h\cdot Ay$ implies that the first column of $A$ lies in a real algebraic set of dimension at most $k(H)$.

In the case where $\Theta = \Aff(V)$, vectors are equivalent if and only if they are equal. As in the real case, if $x \neq y$ we can decompose
$\Aff(x,y)$ into three subsets
\begin{align*} 
S_1&=\{(L,T)\in \Aff(x,y)| \quad x+T=\lambda (y+T), |\lambda| = 1\}\\ 
S_2&=\{(L,T)\in \Aff(x,y)| \quad x+T, (y+T) \text{ are linearly dependent but }|x+T| \neq |y +T|\}\\
S_3&=\{(L,T)\in \Aff(x,y)| \quad x+T \text{ and } y+T \text{ are linearly independent}\}.
\end{align*}
For each $\lambda \in S^1$ there is a unique
translation $T$ such that $(x+T) = \lambda (y+T)$.
Thus, we have that $\dim_\R S_1  = \dim_\R\Aff(V) - \dim_\R(V)+1 \leq \dim_\R\Aff(V) - K+1$.
As in the real case $S_2$ is empty because the action of $H$ is norm preserving.
Likewise, the same argument used in the real case shows that 
$$\dim S_3 \leq \dim_\R \Aff(V) - (\dim_\R(V) - k(H)) = \dim_\R \Aff(V) - K,$$ as in the orthogonal case.

The proof of Theorem~\ref{thm:unitary-unitary} is again very similar to the proof of Theorem~\ref{thm:orthogonal}. The unitary $U(N)$ group
is also a tower of sphere bundles 
$$U(N) = F_N \to F_{N-1} \to \ldots \to F_1 = S^{2N-1},$$
where $F_k$ consists of all unit vectors $(w_1, \ldots , w_k)$, such
that $w_k \in \Span (w_1, \ldots , w_{k-1})^\perp$ and 
$(w_1, \ldots , w_{k-1}) \in F_{k-1}$, and two vectors are equivalent
if $x = \lambda y$ for some $\lambda \in S^1$.

Once again
$U(x,y)$ is empty if $x,y$ are linearly dependent but not equivalent (i.e., they have different norms). If $x,y$ are linearly independent, then using
the notation from the proof of Theorem~\ref{thm:orthogonal}, if
$A \in U(x,y)$ then the second column $w_2$ of $A$ lies in the real algebraic subset $-ab^{-1} w_1 + b^{-1}Hw_1$ for some constants
$a, b$ with $a \neq 0$ and here $w_1$ is the first column. This means that given $w_1$,  $w_2$ is taken from a real algebraic subset of $U(N)$ of dimension $k(H)$, which
has real codimension at least $2N-3 - k(H) = K-3$ in the set of unit vectors $w_2$ orthogonal to $w_1$. Once again we do not impose any further conditions on the columns $A$.

If $H$ is connected, then the requirement that
$$ 1= \langle w_2, w_2 \rangle = |a|^2|b|^{-2} + |b|^{-2}  -2\real( a|b|^{-2} \langle w_1, hw_1)\rangle,$$
imposes at least one additional real algebraic condition on $hw_1$ 
unless $a=0$ and $|b|=1$. Thus, if $x,y$ are not orthogonal
vectors with the same magnitude, then the condition that $w_2$ is a unit
vector shows that $w_2$ is taken from a real algebraic subset of
$S^{2N-3}$ of dimension at most $k(H) -1$, so in this 
case $\dim U(x,y) \leq \dim U(V) - K+2$.

On the other hand, if $x,y$ are orthogonal vectors of the same magnitude, then
after normalizing we conclude that 
$w_1 = x$ and $w_2 = y$ with
$\langle w_1, w_2 \rangle  =0$. If $w_2=y$ is in the $H$-orbit
of $w_1=x$, then writing $w_2 = hw_2$ for some $h \in H$
we see that 
$\langle w_1, h w_1 \rangle=0$. Clearly, this equation cannot
be identically satisfied on
the orbit $Hw_1$, since it does not hold at $w_1 \in Hw_1$, so we impose
an additional condition on the vector $w_2$ since it must also be taken
from the perpendicular space $w_1^\perp$.
Thus, once again we see that $\dim U(x,y) \leq \dim U(V) - K +2$.
\end{proof}

Let $V = \oplus_{\ell =1}^L V_\ell^{\oplus R_\ell}$ be a complex unitary representation of a compact Lie group $G$ where $\dim V_\ell = N_\ell$. 
Consider the problem of signal recovery from the second moment~\eqref{eq:signal_recovery}.
The complex analogue of Corollary \ref{cor:signal_recovery}
is the following. Note that the ambiguity group $H=\prod_{\ell=1}^L U(N_\ell)$ is connected because unitary groups are always connected.

\begin{corollary}\label{cor:signal_recovery_unitary}
Let 
$V = \oplus_{\ell =1}^L V_\ell^{\oplus R_\ell}$ be a complex unitary representation of a compact Lie group $G$.
Let $k = \sum_{\ell =1}^L \dim \U(N_\ell) - \dim \U(N_\ell -R_\ell)$
and set $K = \sum_{\ell =1}^L 2R_\ell N_\ell - k = \dim_\R V -k$.

Then, 
\begin{enumerate}
	\item If $K>M+1$ (resp. $K > M$), then a generic vector $x$ in an
 $\AffV$-generic (resp. $\GLV$-generic) semi-algebraic subset of
 dimension $M$ is determined (resp. up to a scalar multiplication
 by $\lambda \in S^1$) from its second moment.
 	\item  If $K>2M+1$ (resp. $K> 2M$)  then every vector $x$ in an
  $\AffV$-generic (resp. $\GLV$-generic) semi-algebraic subset
  of dimension $M$ is determined (resp. up to a scalar mulitplication by $\lambda \in S^1$) from its second moment.
	\item If  $K>M+2$, then a generic vector $x$ in a $\U(V)$-generic semi-algebraic set of dimension $M$ is determined, up to a scalar multiplication by $\lambda \in S^1$, from its second moment.
 	\item  If $K>2M+2$, then every vector $x$ in a $\U(V)$-generic semi-algebraic set is determined, up to a scalar multiplication by $\lambda \in S^1$, from its second moment.
\end{enumerate}
\end{corollary}

\end{document}